\documentclass[12pt,a4paper]{article}

\usepackage[latin1]{inputenc}
\usepackage[left=2.00cm, right=2.00cm, top=2.00cm]{geometry}
\usepackage{amsmath}
\usepackage{amsthm}
\usepackage{mathtools}
\usepackage{amsfonts}
\usepackage{amssymb}
\usepackage{subcaption}
\usepackage{graphicx}
\usepackage{tikz-cd}

\theoremstyle{plain}
\newtheorem{thm}{Theorem}[section]

\newtheorem{propn}[thm]{Proposition}
\newtheorem{lem}[thm]{Lemma}
\theoremstyle{definition}

\numberwithin{equation}{section}
\newcommand{\C}{\mathbb{C}}
\newcommand{\R}{\mathbb{R}}

\DeclareMathOperator{\Imag}{\text{Im}}

\DeclareMathOperator{\Fix}{\text{Fix}}

\DeclareMathOperator{\Crit}{\text{Crit}}

\DeclareMathOperator{\csch}{\mathrm{csch}}

\newcommand{\CS}{\mathbb{CS}}
\newcommand{\slin}{\mathfrak{sl}}
\title{Unifying the Hyperbolic and Spherical 2-Body Problem with Biquaternions}
\author{Philip Arathoon}
\date{December 2020}
\begin{document}
\maketitle
\begin{abstract}
	The 2-body problem on the sphere and hyperbolic space are both real forms of holomorphic Hamiltonian systems defined on the complex sphere. This admits a natural description in terms of biquaternions and allows us to address questions concerning the hyperbolic system by complexifying it and treating it as the complexification of a spherical system. In this way, results for the 2-body problem on the sphere are readily translated to the hyperbolic case. For instance, we implement this idea to completely classify the relative equilibria for the 2-body problem on hyperbolic 3-space for a strictly attractive potential. 
\end{abstract}
\section*{Background and outline}
The case of the 2-body problem on the 3-sphere has recently been considered by the author in \cite{paper3}. This treatment takes advantage of the fact that \(S^3\) is a group and that the action of \(SO(4)\) on \(S^3\) is generated by the left and right multiplication of \(S^3\) on itself. This allows for a reduction in stages, first reducing by the left multiplication, and then reducing an intermediate space by the residual right-action. An advantage of this reduction-by-stages is that it allows for a fairly straightforward derivation of the relative equilibria solutions: the relative equilibria may first be classified in the intermediate reduced space and then reconstructed on the original phase space.

For the 2-body problem on hyperbolic space the same idea does not apply. Hyperbolic 3-space \(H^3\) cannot be endowed with an isometric group structure and the symmetry group \(SO(1,3)\) does not arise as a direct product of two groups. This prevents us from reducing in stages.

Nevertheless, despite these differences, the sphere and hyperbolic space are both two sides of the same coin. The sphere \(S^3\) is the real affine variety
\begin{equation}
\label{sphere}
u^2+v^2+w^2+z^2=1,
\end{equation}
whereas \(H^3\) is a connected component of 
\begin{equation}\label{1,3}
t^2-x^2-y^2-z^2=1.
\end{equation}
If we complexify the variables these each define the same complex 3-sphere \(\CS^3\) in \(\C^4\). We call \(\CS^3\) the complexification of \(S^3\) and \(H^3\) and refer to \(S^3\) and \(H^3\) as real forms of \(\CS^3\). This is analogous to the notion of real forms and complexifications for vector spaces and Lie groups. To extend this analogy further, one can complexify an analytic Hamiltonian system on a real form to obtain a {holomorphic Hamiltonian system} on the complexification. The original real system appears as an invariant sub-system in the complexified phase space. 

We shall consider the 2-body problem on \(H^3\) and complexify this to obtain a holomorphic Hamiltonian system for the complexified 2-body problem on \(\CS^3\). The group of complex symmetries is the complexification \(SO_4\C\) of \(SO(1,3)\). Pleasingly, this puts us back into the same regime that we had earlier with the \(SO(4)\)-symmetry. In particular, we may take advantage of the \(SL_2\C\times SL_2\C\)-double-cover over \(SO_4\C\) and reduce in stages. We can then apply the exact same analysis for the spherical problem as performed in \cite{paper3}, except now all variables are taken to be complex and one must remember to restrict attention to the invariant hyperbolic real form. This trick allows us to classify the relative equilibria solutions for the hyperbolic problem using the same methods which are used for the spherical case. 

One-parameter subgroups of \(SO(1,3)\) come in four types: elliptic, hyperbolic, loxodromic, and parabolic. The classification of relative equilibria for the 2-body problem on the Lobachevsky plane \(H^2\) was carried out in \cite{jakkuneg,hypisgroup,james}. For each configuration of the two bodies there exists precisely one elliptic and hyperbolic solution up to time-reversal symmetry. We extend this result and show that for each configuration of two bodies in \(H^3\) there exists, up to conjugacy, a circle of loxodromic relative equilibria, including the elliptic and hyperbolic solutions from the \(H^2\) case. There are no parabolic solutions for a strictly attractive potential. Furthermore, we piggy-back on the stability analysis performed in \cite{james} and apply a continuity argument to derive complete results for the reduced stability of the relative equilibria.

The paper is arranged as follows. We first provide a short review of the theory of real forms for holomorphic Hamiltonian systems. The material is distilled from the more complete treatment given in \cite{complex}, but is provided to ensure that the paper remains for the most part self-contained. The complexified 2-body problem is then formulated and conveniently presented in terms of biquaternions. Next, we consider in closer detail the hyperbolic real form. We discuss the hyperbolic symmetry, momentum, and one-parameter subgroups. Finally, we emulate the work in \cite{paper3} and completely classify the hyperbolic relative equilibria by working on an intermediate reduced space for the complexified system. We then derive stability results in the full reduced space.
\section{Real forms for holomorphic Hamiltonian systems}

\subsection{Real-symplectic forms}
A \emph{holomorphic symplectic manifold} is a complex manifold \(M\) equipped with a closed, non-degenerate holomorphic 2-form \(\Omega\). Examples of such manifolds include cotangent bundles of complex manifolds, and coadjoint orbits of complex Lie algebras. Exactly as with real Hamiltonian dynamics we can define a Hamiltonian vector field for a holomorphic function on \(M\) and consider the dynamical system generated by its flow.

A \emph{real structure} \(R\) on a complex manifold \(M\) is an involution whose derivative is conjugate-linear everywhere. If \((M,\Omega)\) is a holomorphic symplectic manifold then a real structure \(R\) is called a \emph{real-symplectic structure} if it satisfies
\[
R^*\Omega=\overline{\Omega}.
\]
If the fixed-point set \(M^R\) is non-empty, then the restriction of \(\Omega\) to \(M^R\) defines a real symplectic form \(\widehat{\omega}_R\). A real structure \(r\) on a complex manifold \(C\) can be lifted to a real-symplectic structure \(R\) on the cotangent bundle by setting
\begin{equation}
\label{cot_lift}
\langle R(\eta),X\rangle=\overline{\langle \eta,r_*X\rangle}
\end{equation}
for \(\eta\in T_x^*C\) and for all \(X\in T_{r(x)}C\). If \(C^r\) is non-empty, then \((T^*C)^R\) is canonically symplectomorphic to the \(T^*C^r\).

\begin{thm}[\cite{complex}]\label{main_thm}
	Let \(f\) be a holomorphic function on a holomorphic symplectic manifold \((M,\Omega)\) equipped with a real-symplectic structure \(R\) with non-empty fixed-point set \(M^R\). If \(f\) is purely real on \(M^R\) then the Hamiltonian flow generated by \(f\) leaves \(M^R\) invariant. Moreover, the flow on \(M^R\) is identically the Hamiltonian flow on \((M^R,\widehat{\omega}_R)\) generated by the restriction \(u\) of \(f\).
\end{thm}

In the situation described by this theorem we may refer to \((M^R,\widehat{\omega}_R,u)\) as a `real form' of the holomorphic Hamiltonian system \((M,\Omega,f)\).

\subsection{Compatible group actions}

Let \(M\) be a complex manifold and \(G\) a complex Lie group acting holomorphically on \(M\). If \(M\) is equipped with a real structure \(R\) then the action is called \emph{$R$-compatible} if there exists a real Lie group structure \(\rho\) on \(G\) which satisfies
\begin{equation}
\label{compatible_condition}
R(g\cdot x)=\rho(g)\cdot R(x)
\end{equation}
for all \(x\) in \(M\) and \(g\) in \(G\). Recall that a real Lie group structure is a real structure on \(G\) which is also a group homomorphism. For such a compatible group action the real subgroup \(G^\rho=\Fix\rho\) acts on \(M^R\). One can show (see for instance \cite[Proposition~2.3]{oshea}) that for \(x\) in \(M^R\)
\begin{equation}
\label{intersection}
T_x(G\cdot x\cap M^R)=T_x(G^\rho\cdot x).
\end{equation}
This fact is significant to the study of relative equilibria solutions for a Hamiltonian system.

Consider a Hamiltonian system \((M,\Omega,f)\), either real or complex, and suppose it admits a symplectic group action by \(G\) which preserves the Hamiltonian. The tuple \((M,\Omega,f,G)\) is a Hamiltonian \(G\)-system. A solution is called a relative equilibria (RE) if it is contained to a \(G\)-orbit. Equivalently, the solution is the orbit of a one-parameter subgroup of \(G\) \cite{marsden_1992}. If we combine this with \eqref{intersection} and Theorem~\ref{main_thm} we have the following.

\begin{propn}\label{RE_trick}
	Let \((M,\Omega,f,G)\) be a holomorphic Hamiltonian \(G\)-system and suppose the action of \(G\) is \(R\)-compatible with respect to a real-symplectic structure \(R\) with non-empty fixed-point set \(M^R\). If \(f\) is real on \(M^R\) then there is a one-to-one correspondence between RE of \((M,\Omega,f,G)\) which are contained to \(N\) and RE of the real form \((M^R,\widehat{\omega}_R,u,G^\rho)\).
\end{propn}

\subsection{Holomorphic momentum maps}

A holomorphic group action of a complex group \(G\) on a holomorphic symplectic manifold \((M,\Omega)\) is Hamiltonian if it admits a holomorphic momentum map 
\(
J\colon M\rightarrow\mathfrak{g}^*
\)
satisfying
\(
\langle J(x),\xi\rangle=H_\xi(x)
\)
for each \(x\) in \(M\) and for all \(\xi\) in the Lie algebra \(\mathfrak{g}\) of \(G\). Here \(H_\xi\) is a holomorphic function whose Hamiltonian vector field is that generated by \(\xi\). If \(\rho\) is a real Lie group structure on \(G\) and \(R\) a real-symplectic structure on \(M\) then the momentum map \(J\) will be called \emph{$R$-compatible} with respect to \(\rho\) if 
\begin{equation}
\label{mu_compatible}
J\circ R=\overline{\rho^*}\circ J.
\end{equation}
The map \(\overline{\rho^*}\) is the \emph{conjugate-adjoint} to \(\rho_*=D\rho_e\) defined by
\begin{equation}
\label{conj_adjoint}
\langle\overline{\rho^*}\eta,\xi\rangle=\overline{\langle\eta,\rho_*\xi\rangle}
\end{equation}
for each \(\eta\) in \(\mathfrak{g}^*\) and for all \(\xi\) in \(\mathfrak{g}\). 

If \(G\) is connected, the \(R\)-compatibility of \(J\) implies the action of \(G\) on \(M\) is also \(R\)-compatible in the sense of \eqref{compatible_condition}. In this case, the symplectic action of \(G^\rho\) on \(M^R\) is also Hamiltonian. Indeed, observe that \(J(x)\) belongs to \(\Fix\overline{\rho^*}\) for \(x\) in \(M^R\). This set consists of all complex-linear forms on \(\mathfrak{g}\) which are real on \(\mathfrak{g}^\rho=\Fix\rho_*\). We therefore have an identification of \(\Fix\overline{\rho^*}\) with \(\left(\mathfrak{g}^\rho\right)^*\). In this way the restriction
\begin{equation}\label{real_momentum}
\widehat{J}\coloneqq J|_{M^R}\colon M^R\longrightarrow\Fix\overline{\rho^*}\cong\left(\mathfrak{g}^\rho\right)^*
\end{equation}
is the momentum map for the action of \(G^\rho\) on \((M^R,\widehat{\omega}_R)\).

\section{The complexified 2-body problem}
\subsection{Biquaternions}
Let \(\{1,I,J,K\}\) denote the standard basis for the real algebra of quaternions. A biquaternion is a linear combination
\begin{equation}\label{biq}
q=u1+vI+wJ+zK
\end{equation}
where \(u,v,w,z\) each belong to \(\C\). As a complex algebra this admits a representation by identifying \(q\) with the matrix
\begin{equation}
\label{matrix}
Q=\begin{pmatrix}
u+iv & w+iz\\ -w+iz & u-iv
\end{pmatrix}.
\end{equation}
The biquaternions are a composition algebra over \(\C\), meaning there exists a complex quadratic form \(||~\cdot~||^2\) which satisfies \(||q_1q_2||^2=||q||^2||q_2||^2\) for all biquaternions \(q_1\) and \(q_2\). With respect to the matrix representation this quadratic form corresponds to the determinant. Observe that the determinant of \(Q\) is 
\[u^2+v^2+w^2+z^2.\]
We may therefore identify the biquaternions with matrices \(Q\) in \(M_2(\C)\) whose quadratic form is \(\det Q\), or with vectors \(q=(u,v,w,z)^T\) in \(\C^4\) whose quadratic form is the standard \(q^Tq\). We will not always specify whether a biquaternion \(q\) is understood to mean an element of \(M_2(\C)\) or \(\C^4\). In most cases it should be clear from the context or not make any difference. In any case, using the polarization identity we shall denote the symmetric, complex bilinear form which gives rise to \(||~\cdot~||^2\) by \(\langle~,~\rangle\).

The set of biquaternions \(q\) with \(||q||^2=1\) forms a group under multiplication. With respect to the matrix representation this corresponds to the special linear group \(SL_2\C\) of matrices with \(\det Q=1\). On the other hand, it also corresponds to the affine variety in \(\C^4\) of vectors satisfying \(q^Tq=1\). This is the complex 3-sphere \(\CS^3\) defined in \eqref{sphere}. Multiplication on the left and right by \(SL_2\C\) on \(M_2(\C)\) preserves the determinant. This provides an action of \(SL_2\C\times SL_2\C\) on \(M_2(\C)\) which we shall denote by
\begin{equation}
\label{lr_action}
(g_1,g_2)\cdot Q=g_1Qg_2^{-1}.
\end{equation}
This action preserves the complex quadratic form on \(\C^4\cong M_2(\C)\), and consequently, establishes the well-known double cover of \(SL_2\C\times SL_2\C\) over \(SO_4\C\). We note that this action preserves \(\CS^3\). Indeed, since \(\CS^3\) may be identified with the group \(SL_2\C\), this action is just left and right multiplication of the group on itself. 
\subsection{The hyperbolic real structure}
We shall be interested in the following real structure \(r\) defined on \(\CS^3\)
\begin{equation}\label{defn_r}
r\colon (u,v,w,z)\longmapsto(\overline{u},-\overline{v},-\overline{w},-\overline{z}).
\end{equation}
According to the matrix representation in \eqref{matrix} this real structure corresponds to taking the conjugate-transpose of \(Q\). The fixed-point set is therefore the set of Hermitian \(2\times 2\) matrices with unit determinant. Alternatively, the fixed-point set consists of the points \((t,ix,iy,iz)\) for \(t,x,y,z\) real numbers satisfying \eqref{1,3}. The solution set has two connected components corresponding to \(t\ge 1\) and \(t\le -1\). We shall only be interested in the \(t\ge 1\) component and write this as \(H^3\), noting that this is the hyperboloid model of hyperbolic-3 space.

The action of \(SL_2\C\times SL_2\C\) in \eqref{lr_action} is \(r\)-compatible with respect to the real Lie group structure
\begin{equation}\label{defn_rho}
\rho(g_1,g_2)=(g_2^{-\dagger},g_1^{-\dagger}).
\end{equation}
The fixed-point set of \(\rho\) is a conjugate-diagonal copy of \(SL_2\C\) which acts by
\begin{equation}\label{hyperbolic_action}
g\cdot Q=gQg^\dagger.
\end{equation}
The hyperboloid \(H^3\subset\R^4\) is invariant with respect to this action. Therefore, the action preserves the indefinite orthogonal form of signature \((1,3)\). Incidentally, this establishes the double cover of \(SL_2\C\) over \(SO(1,3)\).

\subsection{The problem setting}
Consider the holomorphic symplectic manifold \(T^*\CS^3\). Thanks to the complex bilinear form on \(\C^4\) the complex tangent spaces of \(\CS^3\) may be identified with their duals. In this way the cotangent bundle \(T^*\CS^3\) may be identified with the set
\begin{equation}
\label{cot_bundle_complex}
\left\{(q,p)\in \C^4\times\C^4~|~\langle q,p\rangle=0,~||q||^2=1
\right\}.
\end{equation}
The phase space for the complexified 2-body problem is \(T^*\CS^3\times T^*\CS^3\). Strictly speaking the diagonal collision set should be removed, however we will not reflect this in our notation. We consider a holomorphic Hamiltonian
\begin{equation}
\label{ham_complex}
H(q_1,p_1,q_2,p_2)=-\frac{||p_1||^2}{2m_1}-\frac{||p_1||^2}{2m_1}+V\left(\langle q_1,q_2\rangle\right)
\end{equation}
for \(V\colon\C\rightarrow\C\) some holomorphic function which we call the \emph{potential}. 

We now take the product of the real structure \(r\) on \(\CS^3\times\CS^3\) and lift this to a real-symplectic structure \(R\) on the cotangent bundle. The hyperbolic phase space \(T^*H^3\times T^*H^3\) may be identified with a connected component of \(\Fix R\). Restricting the Hamiltonian to this component gives
\begin{equation}
\label{ham_hyperbolic}
H(q_1,p_1,q_2,p_2)=\frac{|p_1|^2}{2m_1}+\frac{|p_1|^2}{2m_1}+V\left(\cosh\psi\right).
\end{equation}
Here \(|p|^2\) denotes the modulus of the covector \(p\in T_q^*H^3\) with respect to the hyperbolic metric on \(H^3\). Recall that this metric is inherited by the negative of the ambient Minkowski metric of signature \((1,3)\) on \(\R^4\). This explains the rather unfortunate presence of what appears to be negative kinetic energy in the holomorphic Hamiltonian. 

For \(q_1\) and \(q_2\) belonging to \(H^3\subset\C^4\) the inner product \(\langle q_1,q_2\rangle\) is equal to \(\cosh\psi\) where \(\psi\) is the hyperbolic distance between \(q_1\) and \(q_2\). For the case of gravitational attraction we choose
\begin{equation}
\label{potential}
V(z)=-m_1m_2\frac{z}{\sqrt{z^2-1}}
\end{equation}
which corresponds to the potential \(-m_1m_2\coth\psi\). In this case the Hamiltonian \(H\) is purely real on the real-symplectic form \(\Fix R\) and so we may apply Theorem~\ref{main_thm} to the holomorphic Hamiltonian system on \(T^*\CS^3\times T^*\CS^3\). It follows that the Hamiltonian system for the 2-body problem on hyperbolic 3-space occurs as a real form of the holomorphic Hamiltonian system described above.
\subsection{Symmetry and momentum}
The momentum map for the cotangent lift of left-multiplication on a group is right-translation of a covector to the identity \cite{bigstages}. Since \(\CS^3\) may be identified with the group \(SL_2\C\) it follows that the momentum map for the action of left-multiplication on phase space is the \emph{total left-momenta}
\begin{equation}
L_{tot}=L_1+L_2
\end{equation}
where \(L_k=p_kq_k^{-1}\) denotes the \emph{left-momentum} of the \(k\)\textsuperscript{th}-particle. Likewise, for the case of right-multiplication the momentum is left-translation to the identity. The momentum map for right-multiplication on phase space is the \emph{total right-momenta}
\begin{equation}
R_{tot}=R_1+R_2
\end{equation}
where \(R_k=q_k^{-1}p_k\) denotes the \emph{right-momentum} of the \(k\)\textsuperscript{th}-particle. The momentum map for the full group action of \(SL_2\C\times SL_2\C\) is therefore
\begin{equation}
\label{momentum_map}
J\colon T^*\CS^3\times T^*\CS^3\longrightarrow\slin_2\C^*\times\slin_2\C^*;\quad(q_1,p_2,q_2,p_2)\longmapsto(L_{tot},-R_{tot}).
\end{equation}
Notice that in the second factor we have negated the total right-momenta. This is because for the action in \eqref{lr_action} we have the inverse of right-multiplication.

By once again using the complex bilinear form on the biquaternions to identify \(\slin_2\C\subset M_2(\C)\) with its dual, we may express the conjugate-adjoint \(\overline{\rho^*}\) of the real Lie group structure \(\rho\) as
\begin{equation}
\overline{\rho^*}(L,R)=(-R^{\dagger},-L^\dagger).
\end{equation} 
The cotangent lift of \(r\) to \(T^*\CS^3\) is given by \((q,p)\mapsto(q^\dagger,p^\dagger)\), where \(q\) and \(p\) are here representing matrices in \(M_2(\C)\). It follows that the momentum map \(J\) is \(R\)-compatible with respect to \(\rho\), and hence, the group action is \(R\)-compatible. Furthermore, according to \eqref{real_momentum} the momentum \(C_{tot}\) for the action of \(SL_2\C\) on \(T^*H^3\times T^*H^3\) may be identified with \(L_{tot}=R_{tot}^\dagger\).

\section{Symmetry and reduction}
\subsection{One-parameter subgroups}
A non-zero element of the Lie algebra \(\slin_2\C\) is conjugate up to the adjoint action to either a semisimple or nilpotent matrix of the form
\begin{equation}
\label{SN}
S=\begin{pmatrix} \eta & 0\\0 & -\eta\end{pmatrix}
\quad\text{or}\quad N=\begin{pmatrix}0 & 1 \\0 & 0\end{pmatrix}
\end{equation}
for \(\eta\) a non-zero complex number. The one-parameter subgroup generated by an element conjugate to \(S\) is called \emph{elliptic} for \(\eta\) imaginary, \emph{hyperbolic} for \(\eta\) real, and \emph{loxodromic} for \(\eta\) a general complex number. A one-parameter subgroup generated by an element conjugate to \(N\) is called \emph{parabolic}. This terminology is inherited from the group \(PSL_2\C\) of M\"{o}bius transformations.

Recall that \((t,ix,iy,iz)\) in \(H^3\) corresponds via \eqref{matrix} to the Hermitian matrix
\begin{equation}
\label{Hermitian_matrix}
\begin{pmatrix}
t-x & iy-z \\-iy-z & t+x
\end{pmatrix}.
\end{equation}
By taking the exponential of the elements \(S\) and \(N\) we may use \eqref{hyperbolic_action} to compute the action of the one-parameter subgroups on \(H^3\). To help visualise these orbits we can use the Poincar\'{e} ball model of hyperbolic 3-space to identify \(H^3\) with the open ball in \(\R^3\) via the map which sends \((t,ix,iy,iz)\) to
\begin{equation*}
(X,Y,Z)=\left(\frac{x}{1+t},\frac{y}{1+t},\frac{z}{1+t}\right).
\end{equation*}
In Figure~\ref{balls} we include an illustration of these orbits on the Poincar\'{e} ball. In these diagrams the \(X\)-axis is the vertical line through the north and south poles.
\begin{figure}
	\centering
	\begin{subfigure}{0.32\textwidth}
		\centering
		\includegraphics[scale=.6]{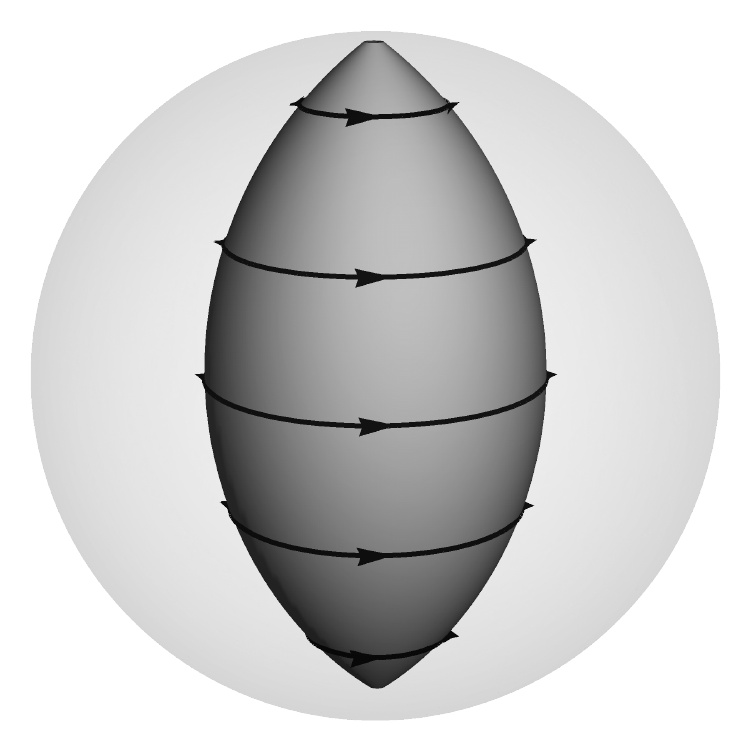}
		\caption{Elliptic}
	\end{subfigure}
	\begin{subfigure}{0.32\textwidth}
		\centering
		\includegraphics[scale=.6]{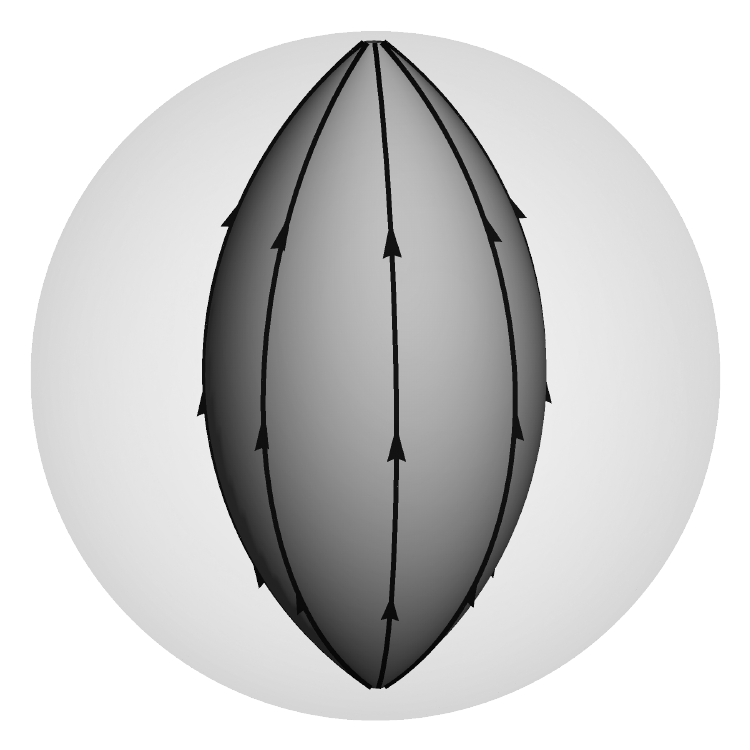}
		\caption{Hyperbolic}
	\end{subfigure}
	\\
	\begin{subfigure}{0.32\textwidth}
		\centering
		\includegraphics[scale=.6]{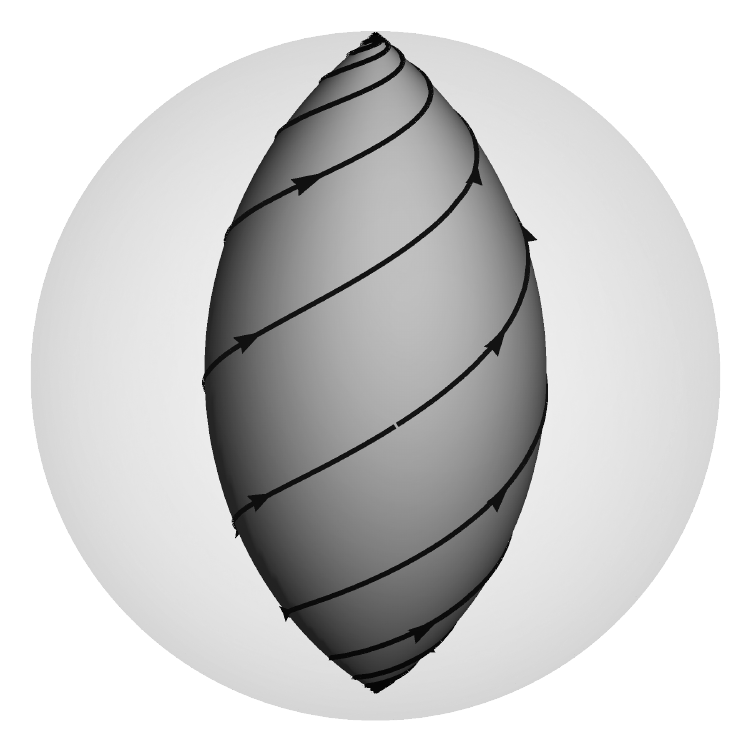}
		\caption{Loxodromic}
	\end{subfigure}
	\begin{subfigure}{0.32\textwidth}
		\centering
		\includegraphics[scale=.6]{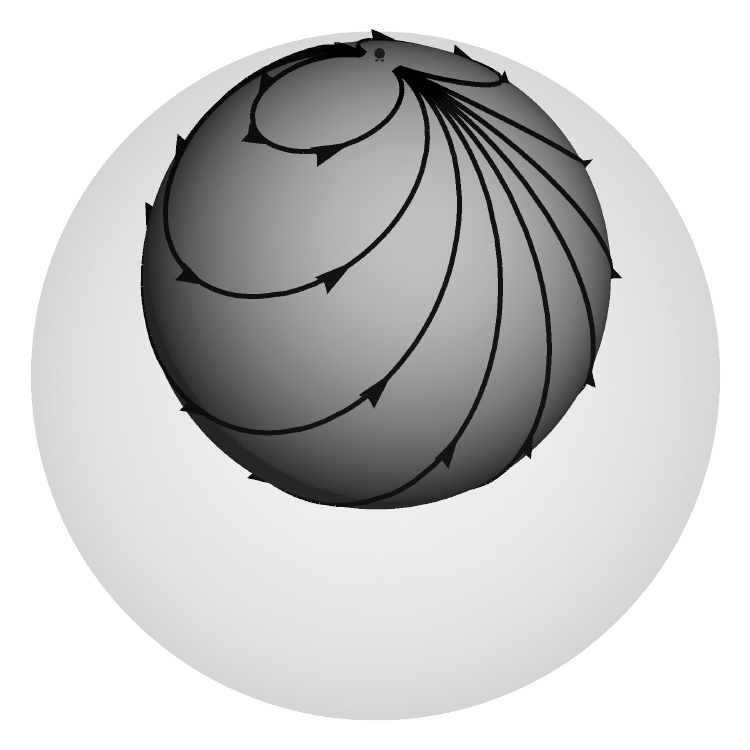}
		\caption{Parabolic}
	\end{subfigure}
	\caption{\label{balls} Orbits of one-parameter subgroups on the Poincar\'{e} ball. }
\end{figure}
\subsection{Geodesics and the Lobachevsky plane}
Let \(H^2\) denote the submanifold of \(H^3\) given by setting \(y=0\). Note from \eqref{Hermitian_matrix} that this submanifold equivalently corresponds to those purely real, symmetric matrices with unit determinant. It also corresponds to the open disk in the Poincar\'{e} ball by setting \(Y=0\). More generally, we shall refer to any 2-dimensional submanifold of \(H^3\) as a \emph{Lobachevsky plane} if it can be transformed to \(H^2\) by the action of \(SL_2\C\).

If we set \(y=0\) and \(x=0\) we obtain the curve \(H^1\). In terms of the general form of a biquaternion given in \eqref{biq} such elements of \(H^1\) may be written as
\begin{equation}\label{split}
e^{\psi j}=\cosh\psi+j\sinh\psi
\end{equation}
where \(j=-iK\). Since \(j^2=1\) we see that \(H^1\) is none other than the unit hyperbola in the algebra of split-complex numbers; the analogue of the unit circle in the complex numbers. Finally, we remark that \(H^1\) is a geodesic in \(H^3\) and that the action of \(SL_2\C\) can send any geodesic to \(H^1\).

\begin{propn}\label{momentum_crit}
	Let \(\widehat{J}\) denote the momentum map for the action of \(SL_2\C\) on \(T^*H^3\times T^*H^3\). The critical points of \(\widehat{J}\) are those points \((q_1,p_1,q_2,p_2)\) where \(p_1\) and \(p_2\) are each tangent to the geodesic through \(q_1\) and \(q_2\). The corresponding momentum \(C_{tot}\) in \(\slin_2\C^*\) has \(||C_{tot}||^2\) a non-negative real number. Conversely, \(||C_{tot}||^2\) is real if and only if \(p_1\) and \(p_2\) are each tangent to a Lobachevsky plane containing \(q_1\) and \(q_2\).
\end{propn}
\begin{proof}
	A point is a critical point of the momentum map if and only if the action at this point is not locally free. We may suppose \(q_1\) and \(q_2\) belong to \(H^1\). The isotropy subgroup fixing \(H^1\) is the group of rotations \(SO(2)\subset SO(1,3)\) in the \(xy\)-plane. Therefore, the action fails to be locally free when \(p_1\) and \(p_2\) have no \(x\) or \(y\) component. A calculation in terms of the split-complex numbers reveals that \(||L_{tot}||^2\) is real and non-negative. 
	
	We may use the \(SL_2\C\)-action to place the first particle at the centre of the Poincar\'{e} ball where \(q_1=I\). We can then rotate the ball so that \(q_2\) belongs to \(H^1\) and \(p_2\) to \(H^2\). Since \(p_1,q_2\) and \(p_2\) are all Hermitian, the imaginary part of \(||C_{tot}||^2\) is 
	\[
	\det(L_1+L_2)-\det(L_1^\dagger+L_2^\dagger)=2\langle p_1,[p_2,q_2^{-1}]\rangle.
	\]
	Observe that \(p_1\) is a traceless, Hermitian matrix since it is tangent to \(H^3\) at \(q_1=I\), and that \([p_2,q_2^{-1}]\) is a real, skew-symmetric matrix as \(p_2\) and \(q_2\) are each real and symmetric. It follows that \(||C_{tot}||^2\) is real if and only if, either \(p_2\) commutes with \(q_2\), or if \(p_1\) is real. If \(p_1\) is real then everything belongs to \(H^2\). If \(q_2\) commutes with \(p_2\), then \(p_2\) must also be a split-complex number, and hence, \(p_2\) is tangent to \(H^1\) at \(q_1\). We can therefore rotate the ball about the \(H^1\)-axis to send \(p_1\) into \(H^2\).
\end{proof}

\subsection{The right-reduced space}
Instead of reducing by the full \(SL_2\C\times SL_2\C\)-symmetry we shall instead obtain an intermediate reduced space by reducing by the action of right-multiplication by \(SL_2\C\). The orbit quotient is given by sending \((q_1,p_1,q_2,p_2)\) to 
\[
(L_1,L_2,q_R)
\]
where we have introduced the right-invariant quantity \(q_R=q_1q_2^{-1}\). One can directly perform this orbit quotient by using Corollary~3.8.5 in \cite{marsden_1992} or by applying the Semidirect Product Reduction by Stages Theorem as in \cite{paper3}. The reduced Poisson structure on \(\slin_2\C^*\times\slin_2\C^*\times\C^4\) is the Poisson structure on the dual of the complexified special-Euclidean Lie algebra \(\mathfrak{se}_4\C=\mathfrak{so}_4\C\ltimes\C^4\). The Hamiltonian \(H\) descends to
\begin{equation}
H(L_1,L_2,q_R)=-\frac{||L_1||^2}{2m_1}-\frac{||L_1||^2}{2m_1}+V(z)
\end{equation}
where \(z=\langle q_1,q_2\rangle=\langle q_R,1\rangle\). We can use the explicit expression for the Lie-Poisson equations for a semidirect product \cite{semidirect} to obtain the equations of motion
\begin{equation}\label{Ham_eqns}
\def\arraystretch{1.5}
\begin{array}{l}
\dot{L}_1=+f(z)\Imag(q_R),\\\dot{L}_2=-f(z)\Imag(q_R),\\
\displaystyle\dot{q}_R=-\frac{L_1}{m_1}q_R+q_R\frac{L_2}{m_2}.
\end{array}
\end{equation}
In this set of equations \(f(z)=dV/dz\) and \(\Imag(q)=q-\langle q,1\rangle\) is the `imaginary part' of a biquaternion \(q\).
\subsection{Hyperbolic relative equilibria}
The action of left-multiplication by \(SL_2\C\) descends to the right-reduced space in \(\mathfrak{se}_4\C^*\) as
\begin{equation}\label{residual}
g\cdot(L_1,L_2,q_R)=\left(gL_1g^{-1},gL_2g^{-1},gq_Rg^{-1}\right).
\end{equation}
A RE of the system on \(T^*\CS^3\times T^*\CS^3\) is a solution contained to an orbit of \(SL_2\C\times SL_2\C\). It must therefore descend to a RE on the right-reduced space with respect to the group action above. Combined with Proposition~\ref{RE_trick} we see that RE of the hyperbolic 2-body problem can be found by classifying those RE in \(\mathfrak{se}_4\C^*\) which lift to solutions on \(T^*H^3\times T^*H^3\).

Without any loss of generality, we may suppose the RE are orbits of one-parameter subgroups generated by \(S\) or \(N\) given in \eqref{SN}. We separate this subsection into two parts to handle both cases separately.

\subsubsection{Semisimple RE}

By differentiating the action of the one-parameter subgroup \(g(t)=\exp(St)\) in \eqref{residual} and setting this to equal the equations of motion in \eqref{Ham_eqns} we obtain
\begin{align}
\left[S,L_1\right]&=+f\Imag(q_R),\label{one}\\
\left[S,L_2\right]&=-f\Imag(q_R),\label{two}\\
\left[S,q_R\right]&=-\frac{L_1}{m_1}q_R+q_R\frac{L_2}{m_2}. \label{three}
\end{align}
The first two equations imply \(\Imag(q_R)\) is orthogonal to \(S=-i\eta I\), where we are now using the biquaternion notation from \eqref{biq}. Therefore, we may suppose \(q_R\) is a biquaternion of the form \(u1+zK\). For hyperbolic RE we must have \(\langle q_R,1\rangle=\langle q_1,q_2\rangle=\cosh\psi\), and so it follows that \(q_R\) is the split-complex number \(e^{\psi j}\) from \eqref{split}. Equations~\eqref{one} and \eqref{two} are now satisfied for 
\begin{equation}
\label{xy}
L_1=x_1I+yJ\quad\text{and}\quad L_2=x_2I-yJ,
\end{equation}
where
\begin{equation}\label{defn_y}
y=\frac{f\sinh\psi}{2\eta},
\end{equation}
and \(x_1\) and \(x_2\) are complex numbers yet to be determined. Expanding Equation~\eqref{three} using the multiplication of biquaternions results in a linear system of equations in \(x_1\) and \(x_2\). For \(\psi\ne 0\) the solutions are uniquely given by
\begin{equation}\label{x1x2}
\def\arraystretch{2.5}
\begin{array}{l}
\displaystyle x_1=iy\left(\coth 2\psi+\frac{m_1}{m_2}\csch 2\psi\right)+im_1\eta,\\
\displaystyle x_2=iy\left(\coth 2\psi+\frac{m_2}{m_1}\csch 2\psi\right)+im_2\eta.
\end{array}
\end{equation}
\subsubsection{Non-existence of parabolic RE}
For a parabolic RE we replace \(S\) with \(N\) in Equations~\eqref{one}--\eqref{three}. By considering the action of the isotropy subgroup of \(N\) in \(SL_2\C\) we may suppose that \(q_R\) is of the form \(\cosh\psi+i\sinh\psi I\). In this case, Equations~\eqref{one} and \eqref{two} are satisfied for \(L_s=y_sJ+z_sK\) where
\[
y_1-iz_1=-y_2+iz_2=f\sinh\psi.
\]
Expanding Equation~\ref{three} reveals that we require
\[
0=f\sinh\psi(m_1e^\psi+m_2e^{-\psi}).
\]
For a strictly attractive potential \(f\) is never zero, and so the equation above cannot hold for \(\psi\ne 0\), and hence, no such parabolic RE exist.
\subsection{Reconstruction and classification}
According to the matrix representation of biquaternions in \eqref{matrix} \(q_R=e^{\psi j}\) is Hermitian. For a hyperbolic RE \(q_1\) and \(q_2\) are also Hermitian, therefore, since \(q_R=q_1q_2^{-1}\) we see that this implies \(q_1\) and \(q_2\) each commute with \(q_R\). This is only possible if they too belong to the split-complex numbers. Therefore, we must have
\(
q_1=e^{\chi_1j}\) and \( q_2=e^{-\chi_2j}
\)
for \(\chi_1\) and \(\chi_2\) real numbers satisfying \(\chi_1+\chi_2=\psi\).

For each particle \(q_s\) we may differentiate the \(SL_2\C\)-action in \eqref{hyperbolic_action} for the orbit of \(g(t)=\exp({St})\) to find the velocity vector 
\(
\dot{q}_s=Sq+qS^\dagger
\).
The momentum \(p_s\) is then given by \(-m_s\dot{q}_s\). Recall that this minus sign is a consequence of the negative kinetic energy term in \eqref{ham_complex}. We can then calculate \(L_1\) and \(L_2\) explicitly and compare this with the expression in \eqref{xy}. From this calculation we obtain
\begin{equation}
\label{x_and_y}
x_s=im_s(\eta+\overline{\eta}\cosh 2\chi_s)\quad\text{and}\quad y=m_s\overline{\eta}\sinh 2\chi_s.
\end{equation}
By setting these expressions to equal those in Equations~\eqref{defn_y} and \eqref{x1x2} we obtain the full classification of RE for the 2-body problem on \(H^3\).
\begin{thm}\label{classification}
	Every semisimple RE for the 2-body problem on \(H^3\) is conjugate to a RE generated by the biquaternion \(S=-i\eta I\) in \(\slin_2\C\subset M_2(\C)\) for a complex number \(\eta\). These solutions are classified up to conjugacy by the \emph{separation} \(\psi>0\) between the particles and the \emph{phase} \(\theta=\arg\eta\in[0,\pi)\). For each \((\theta, \psi)\) we may suppose \(q_1=e^{\chi_1j}\) and \(q_2=e^{-\chi_2 j}\) for \(\chi_1\) and \(\chi_2\) positive real numbers uniquely determined by \(\chi_1+\chi_2=\psi\) and
	\[
	m_1\sinh 2\chi_1=m_2\sinh 2\chi_2.
	\]
	The modulus of \(\eta\) satisfies
	\[
	|\eta|^2=\frac{f\sinh\psi}{2\zeta}
	\]
	where \(\zeta=m_1\sinh 2\chi_1=m_2\sinh 2\chi_2\). Here \(f=dV/dz\) where \(z=\cosh\psi\) and \(V\) is the potential. For a strictly attractive potential there are no parabolic RE.
\end{thm}
\section{Stability}
\subsection{Some invariant theory}
We may further reduce the right-reduced space by the residual action of left-multiplication in \eqref{residual} to obtain the {full reduced space}. To do this we shall consider the (holomorphic) Poisson reduction of \(\mathfrak{se}_4\C^*\) by the \(SL_2\C\)-action. 

By writing \((L_1,L_2,q_R)\) as \((L_1,L_2,\Imag(q_R),z)\) the action decomposes into three irreducible \(\slin_2\C\)-components and one trivial component. The isomorphism \(\slin_2\C\cong\mathfrak{so}_3\C\) intertwines the \(SL_2\C\)-adjoint-representation with the standard vector representation of \(SO_3\C\). We can therefore employ the First Fundamental Theorem of Invariant Theory to obtain generators of the \(SL_2\C\)-invariant ring, see for instance \cite{fultonharris}. If we let \((v_1,v_2,v_3)\) denote \((L_1,L_2,\Imag(q_R))\) then the generators of this ring are the pairwise products \(k_{ij}=\langle v_i,v_j\rangle\) and the determinant \(\delta=\langle v_3,[v_1,v_2]\rangle\). These satisfy the algebraic relations \(k_{ij}=k_{ji}\) and \(\delta^2=\det k_{ij}\).

The categorical quotient for this action is the map 
\begin{equation}\label{pi}
\pi\colon(L_1,L_2,q_R)\longmapsto\left(\{k_{ij}\}_{i\le j},\delta,z\right)\in\C^8.
\end{equation}
Unfortunately, this is not an orbit-map since there exist orbits which are not separated by the invariants. Geometric Invariant Theory arises to resolve this issue by restricting to a subset of so-called stable points upon which this map is a geometric quotient, or in other words, an orbit-map \cite{newstead}.
\begin{lem}
	If \(v_1\), \(v_2\) and \(v_3\) are each semisimple elements in \(\slin_2\C\), and not all colinear, then \((v_1,v_2,v_3)\) is stable with respect to the diagonal \(SL_2\C\)-action.
\end{lem}
\begin{proof}
	Consider the action of a one-parameter subgroup \(\C^\times\) generated by a semisimple generator \(S\) as in \eqref{SN}. The space \(\slin_2\C\) decomposes as \(V_0\oplus V_{-1}\oplus V_{+1}\) where the action on \(V_w\) is \(t\cdot v=t^wv\). Note that \(V_{+1}\) and \(V_{-1}\) must therefore be isotropic subspaces with respect to the complex bilinear form \(\langle~,~\rangle\). The Hilbert-Mumford criterion tells us that a point is stable if there does not exist any such one-parameter subgroup for which \(v_1\), \(v_2\), and \(v_3\) all belong to either \(V_0\oplus V_{-1}\) or \(V_0\oplus V_{+1}\). Suppose \(v_i\) and \(v_j\) belong to such a subspace. This subspace is degenerate with respect to \(\langle~,~\rangle\) and so \(4(k_{ii}k_{jj}-k_{ij}^2)=\det([v_i,v_j])\) must be zero. By semisimplicity of \(v_i\) and \(v_j\) this implies \([v_i,v_j]=0\), and therefore, \(v_i\) and \(v_j\) are colinear. 
\end{proof}
\subsection{The full reduced space}
Let \(M\) denote the subset of stable points of \(T^*\CS^3\times T^*\CS^3\) with respect to the \(G=SL_2\C\times SL_2\C\)-action. The elements \(L_1\), \(L_2\), and \(\Imag(q_R)\) are always semisimple on the hyperbolic real form. Moreover, they are only ever colinear when the \(SL_2\C\)-action fails to be locally free, and hence, by Proposition~\ref{momentum_crit} we see that the real form \(M^R\) is 
\begin{equation*}\label{MR}
(T^*H^3\times T^*H^3)\setminus\Crit\widehat{J}.
\end{equation*}
The group \(G\) acts freely on \(M^R\), and therefore, thanks to \(R\)-compatibility, the intersection of a \(G\)-orbit with \(M^R\) is an orbit of \(G^\rho\). This implies that \(M^R\hookrightarrow M\) descends to an injection of orbit spaces \(M^R/G^\rho\hookrightarrow M/G\). It follows that \(\pi\) restricted to \(M^R\) defines an orbit map for the action of \(SL_2\C\) on \((T^*H^3\times T^*H^3)\setminus\Crit\widehat{J}\).

In fact, although we shall not pursue this further, we remark that because the action of \(G\) is \(R\)-compatible, \(R\) descends to a real structure \(\widetilde{R}\) on the reduced space \(M/G\). The real reduced space \(M^R/G^\rho\) belongs to the fixed-point set of \(\widetilde{R}\) which we call a \emph{real-Poisson form}. For our purposes, the most important consequence of this is that the symplectic leaves of \(M^R/G^\rho\) are real-symplectic forms of the holomorphic symplectic leaves in \(M/G\) \cite[Theorem~2.2]{complex}.

\subsection{Degenerate RE}
Thanks to the classification in Theorem~\ref{classification} we may parameterise the hyperbolic RE in the full reduced space by \((\theta, \psi)\). This parameterises a surface \(\Fix H\) of fixed-points in the full reduced space. The critical values of the energy-Casimir map
\[
\mu\colon M^R/G^\rho\rightarrow\R\times\C;\quad\left(\{k_{ij}\}_{i\le j},\delta,z\right)\longmapsto(H,||C_{tot}||^2)
\] 
are precisely the image of \(\Fix H\). We call this image the \emph{energy-Casimir diagram}. For the gravitational potential in \eqref{potential} this image is included in Figure~\ref{diagram}. In practice, to obtain this image it is first necessary to show that
\begin{equation}\label{zeta}
\zeta=m_1m_2Z^{-1}e^\psi\sinh 2\psi
\end{equation}
where
\begin{equation}\label{Z_defn}
Z=\sqrt{(m_1+m_2e^{2\psi})(m_2+m_1e^{2\psi})}.
\end{equation}
One can then express \(\chi_1\), \(\chi_2\), and \(|\eta|^2\) in terms of \((\theta, \psi)\) and use \eqref{x_and_y} to carry out the computation. Observe that this image is pinched along a singular cusp. As the next proposition shows, this is indicative of degenerate RE. The proof is immediate.
\begin{figure}
	\centering
	\includegraphics{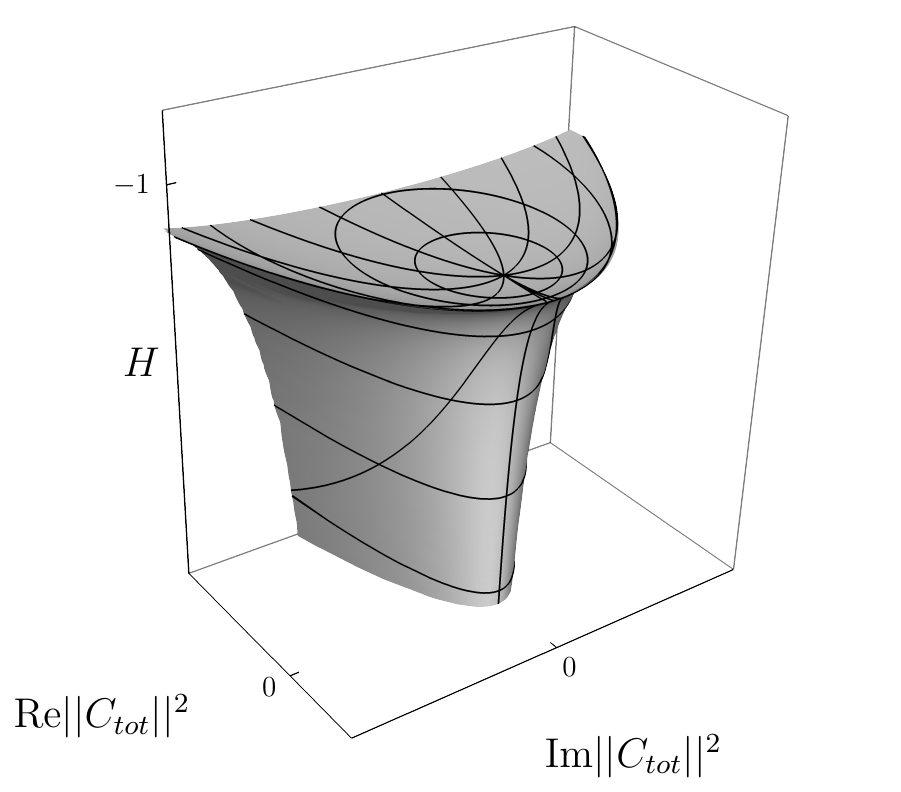}
	\caption{\label{diagram} Energy-Casimir diagram for \(m_1=1\) and \(m_2=2\). The lines emanating from the focal point are curves of constant \(\theta\) and those transversal to them are of constant \(\psi\).}
\end{figure}
\begin{propn}\label{RE_result}
	Let \(P\) be a Poisson manifold and suppose the symplectic leaf \(M\) through \(x\) is the regular level set of a Casimir function \(C\). Consider the Hamiltonian flow generated by \(H\) and suppose \(\Fix H\) is an immersed submanifold containing \(x\). The subspace \(\ker D_x(C|_{\Fix H})\subset T_xM\) is fixed by the linearised flow of \(H\) at \(x\).
\end{propn}
As a corollary we see that the cusp in Figure~\ref{diagram} must correspond to degenerate RE. Indeed, since this cusp consists of singular values of \(\mu\) restricted to \(\Fix H\), they must also be singular values of the Casimir map restricted to \(\Fix H\). 
\begin{figure}
	\centering
	\includegraphics{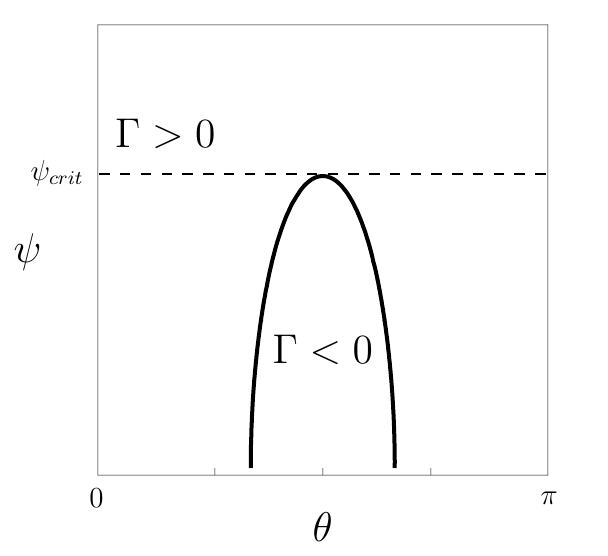}
	\caption{\label{curve} The curve \(\Gamma(\theta, \psi)=0\) of degenerate RE.}
\end{figure}

\begin{propn}
	A RE with separation \(\psi\) and phase \(\theta\) is degenerate if and only if
	\begin{equation}
	\Gamma(\theta, \psi)=
	(m_1+m_2)Z\cos 2\theta+\left(\frac{1+\tanh\psi}{2\cosh\psi}\right)(m_1+m_2\cosh 2\psi)(m_2+m_1\cosh 2\psi)
	\end{equation}
	is equal to zero. These RE correspond to the singular points along the cusp of the energy-Casimir diagram.
\end{propn}
\begin{proof}
	Thanks to our formulation of the problem as the complexification of the 2-body problem on the sphere, the fully reduced equations of motion in \(\C^8\) coincide with those in Equation~2.10 of \cite{complex}. These may be linearised at a fixed-point in a symplectic leaf. A RE is degenerate if and only if the constant term of the characteristic polynomial of this linearisation is zero. This polynomial is derived in Equation~3.16 of \cite{complex}, however for a hyperbolic RE in \(\Fix H\) we must remember to replace every instance of a trigonometric function with its hyperbolic counterpart in \(\psi\). The constant term \(c_0\) is then given by
	\[\left(\frac{k_{11}}{m_1^2}-\frac{k_{22}}{m_2^2}\right)^2+2\coth\theta\csch^2\theta\left[\frac{k_{11}}{m_1}\left(1+\frac{m_2}{m_1}\right)+\frac{k_{22}}{m_2}\left(1+\frac{m_1}{m_2}\right)\right]+\left[(m_1+m_2)\coth\theta\csch^2\theta\right]^2.\]
	The expressions in \eqref{x_and_y} allow us to write 
	\[
	k_{11}=-2|\eta|^2m_1^2(\cos 2\theta+\cosh 2\chi_1),
	\]
	which can then be rewritten using \eqref{zeta} and \eqref{Z_defn} as
	\[
	k_{11}=-\frac{m_1^2}{2\sinh^3\psi\cosh\psi}\left[Ze^{-\psi}\cos 2\theta+\left(m_1\cosh 2\psi+m_2\right)\right].
	\]
	A similar expression holds for \(k_{22}\). If we substitute this into the constant term above we find that \(c_0=-2e^{-\psi}\Gamma(\theta,\psi)\csch^6\psi\). In Figure~\ref{curve} we illustrate the solutions to \(\Gamma(\theta,\psi)=0\). As this defines a connected curve in \((\theta,\psi)\)-space it must correspond exactly to the degenerate RE along the singular cusp in Figure~\ref{diagram}.
\end{proof}
	\subsection{The Energy-Casimir method}
If \(C=||C_{tot}||^2\) is real then by Proposition~\ref{momentum_crit} the orbit-reduced space may be identified with the space of \(SL_2\C\)-orbits which intersect
\begin{equation}
\label{set}
\left\{
(q_1,p_1,q_2,p_2)\in T^*H^2\times T^*H^2~|~{C}_{tot}\in\slin_2\R, ~||{C}_{tot}||^2=C
\right\}.
\end{equation}
	Observe that \(\{q_1,p_1,q_2,p_2\}\) spans the space of real, symmetric matrices in \(M_2(\C)\). From \eqref{hyperbolic_action} we see that the orbits contained within this space are equivalently the orbits of the subgroup \(SL_2\R\). Therefore, thanks to the isomorphism \(PSL_2\R\cong SO(1,2)\) the reduced space for \(C=||C_{tot}||^2\) real is identical to the orbit-reduced space for the 2-body problem on \(H^2\) with Casimir \(C\).
	
	The Hessian and linearisation of the reduced Hamiltonian for the problem on \(H^2\) is found in \cite{james}. We can therefore apply a continuity argument to extend their results to \(H^3\), noting that the behaviour of the fixed-points can only change when they cross a degenerate point whereupon an eigenvalue becomes zero.
	\begin{thm}
		The stability of a RE for the gravitational 2-body problem on \(H^3\) is determined by \(\Gamma(\theta, \psi)\) as follows:
		\begin{enumerate}
			\item For \(\Gamma<0\) the Hessian of the Hamiltonian at the fixed-point in reduced space is positive definite, and thus, the RE is Lyapunov stable.
			\item For \(\Gamma>0\) the Hessian at the fixed-point in reduced space has signature \((+++-)\), and thus, the RE is linearly unstable.
		\end{enumerate}
	In particular, a RE is unstable whenever  \(\cos 2\theta>0\), or whenever \(\psi>\psi_{crit}\).
	\end{thm}

%%%%ADDRESSES
\newcommand{\Addresses}{{% additional braces for segregating \footnotesize
		\bigskip
		\footnotesize
		
		P.~Arathoon, \textsc{University of Manchester}\par\nopagebreak
		\texttt{philip.arathoon@manchester.ac.uk}
	}}
%		
%		%		\medskip
%		%		
%		%		M.~Dane (Corresponding author), \textsc{Atmospheric Research Station,
%		%			Pala Lundi, Fiji}\par\nopagebreak
%		%		\textit{E-mail address}, M.~Dane: \texttt{DaneMark@@ffr.choice}
%		%		
%		%		\medskip
%		%		
%		%		J.~Jones, \textsc{Department of Philosophy, Freedman College,
%		%			Periwinkle, Colorado 84320}\par\nopagebreak
%		%		\textit{E-mail address}, J.~Jones: \texttt{id739e@@oseoi44 (Bitnet)}
%		
%}}
%
\Addresses
\end{document}